\documentclass[
    aps,
    pra,
    twocolumn,
    superscriptaddress,
    amsmath,
    amssymb,
    floatfix
    ]{revtex4-1} 
\usepackage{graphicx} 
\usepackage{hyperref} 
\usepackage{braket}
\usepackage{mathtools}
\usepackage{bm}
\usepackage[normalem]{ulem}
\usepackage{xcolor}
\usepackage{amsthm}

\newcommand{\QAOA}{\text{QAOA}^2}
\newcommand{\bz}{\bm{z}}
\newcommand{\bx}{\bm{x}}
\newcommand{\bs}{\bm{s}}
\newcommand{\bgamma}{\bm{\gamma}}
\newcommand{\bbeta}{\bm{\beta}}
\newcommand{\bbx}{\bar{\bm{x}}}
\newcommand{\bW}{W}

\newtheorem{thm}{Theorem}

\usepackage{mathtools}

\begin{document}

\title{QAOA-in-QAOA: solving large-scale MaxCut problems on small quantum machines}

\author{Zeqiao Zhou}
\thanks{This work was done when he was a research intern at JD Explore Academy, zhouzeqiao@mail.ustc.edu.cn}
\affiliation{Department of Electronic Engineering and Information Science, University of Science and Technology of China, Hefei, Anhui, China, 230027}
\affiliation{JD Explore Academy, Beijing 101111, China}

\author{Yuxuan Du}
\thanks{Corresponding author, duyuxuan123@gmail.com}
\affiliation{JD Explore Academy, Beijing 101111, China}

\author{Xinmei Tian}
\thanks{Corresponding author, xinmei@ustc.edu.cn}
\affiliation{Department of Electronic Engineering and Information Science, University of Science and Technology of China, Hefei, Anhui, China, 230027}
 
\author{Dacheng Tao}
\thanks{Corresponding author, dacheng.tao@sydney.edu.au}
\affiliation{JD Explore Academy, Beijing 101111, China}

\date{\today}

\begin{abstract}
The design of fast algorithms for combinatorial optimization greatly contributes to a plethora of domains such as logistics, finance, and chemistry.  Quantum approximate optimization algorithms (QAOAs), which utilize the power of quantum machines and inherit the spirit of adiabatic evolution, are novel approaches to tackle combinatorial problems with potential runtime speedups. However, hurdled by the limited quantum resources nowadays, QAOAs are infeasible to manipulate large-scale problems. To address this issue, here we revisit the MaxCut problem via the divide-and-conquer heuristic: seek the solutions of subgraphs in parallel and then merge these solutions to obtain the global solution. Due to the $\mathbb{Z}_2$ symmetry in MaxCut, we prove that the merging process can be further cast into a new MaxCut problem and thus be addressed by QAOAs or other MaxCut solvers. With this regard, we propose QAOA-in-QAOA ($\QAOA$) to solve arbitrary large-scale MaxCut problems using small quantum machines. We also prove that the approximation ratio of $\QAOA$ is lower bounded by $1/2$. Experiment results illustrate that under different graph settings, $\QAOA$ attains a competitive or even better performance over the best known classical algorithms when the node count is around $2000$. Our method can be seamlessly embedded into other advanced strategies to enhance the capability of QAOAs in large-scale combinatorial optimization problems.
\end{abstract}

\maketitle

\section{Introduction}
Combinatorial optimization \cite{korteCombinatorialOptimizationTheory2006}, which aims to search for maxima or minima of an objective function with discrete solution space, is an indispensable tool in various application domains such as portfolio investment, vehicle routing and transportation \cite{juarnaCombinatorialAlgorithmsPortfolio2017,sbihiCombinatorialOptimizationGreen2007}. Driven by its fundamental importance, huge efforts have been dedicated to devising efficient algorithms for combinatorial problems during past decades. Representative examples include using semidefinite programming techniques to approximate the solution of MaxCut and maximum 2-satisfiability problems \cite{goemans879approximationAlgorithmsMAX1994,goemansImprovedApproximationAlgorithms1995},  adopting the simulated annealing methods to solve constrained problems   \cite{kirkpatrickOptimizationSimulatedAnnealing1983},  and exploiting other heuristics such as expert knowledge and the structure of studied problems to solve traveling salesman problems \cite{johnsonLocalOptimizationTraveling1990,jungerChapterTravelingSalesman1995,regoTravelingSalesmanProblem2011}. Recently, attributed to the power of neural networks, deep learning techniques have also been employed in solving combinatorial optimization problems  \cite{bengioMachineLearningCombinatorial2020,mazyavkinaReinforcementLearningCombinatorial2021,cappartCombinatorialOptimizationReasoning2021}.  Nevertheless, due to the intrinsic hardness of most combinatorial  problems \cite{karpReducibilityCombinatorialProblems1972}, existing methods request expensive computational overhead to estimate the optimal solution and thus it is highly demanded to design more powerful algorithms to speed up the optimization.

Quantum computers have the ability to efficiently solve certain problems that are computationally hard for classical computers \cite{shorAlgorithmsQuantumComputation1994}. This superiority could be preserved for noisy intermediate-scale quantum (NISQ) machines~\cite{preskillQuantumComputingNISQ2018}, because both theoretical and experimental studies have exhibited their runtime merits over the classical counterparts for certain tasks   \cite{zhongQuantumComputationalAdvantage2020,aruteQuantumSupremacyUsing2019,wuStrongQuantumComputational2021}. For this reason, there is a growing interest in designing NISQ algorithms with computational merits. Variational quantum algorithms (VQAs) \cite{cerezoVariationalQuantumAlgorithms2021}, which consist of parameterized quantum circuits \cite{benedettiParameterizedQuantumCircuits2019} and classical optimizers to adequately leverage  accessible quantum resources and suppress the system noise, are leading solutions to achieve this goal. Notably, initial studies have exhibited that quantum approximate optimization algorithms (QAOAs)~\cite{edwardfarhiQuantumApproximateOptimization2014}, as one crucial paradigm of VQAs, can be used to tackle combinatorial optimization with potential computational advantages \cite{farhiQuantumSupremacyQuantum2019a,guerreschiQAOAMaxCutRequires2019}. The underlying principle of QAOAs is mapping a quadratic unconstrained binary optimization (QUBO) problem describing the explored combinatorial problem to a Hamiltonian whose ground state refers to the optimal solution   \cite{lucasIsingFormulationsMany2014,gloverTutorialFormulatingUsing2019}. In this way,  various manipulable quantum systems can be used to advance combinatorial problems \cite{hamerlyExperimentalInvestigationPerformance2019,paganoQuantumApproximateOptimization2020,harriganQuantumApproximateOptimization2021a}.

Envisioned by the promising prospects, both empirical and theoretical studies have been carried out to understand the foundation of QAOAs and improve their performance. One critical line of research is unveiling the connection between adiabatic quantum computation   \cite{farhiQuantumComputationAdiabatic2000,farhiQuantumAdiabaticEvolution2001} and QAOAs and showing that QAOAs can be seen as a parameterized Trotterization of adiabatic evolution~\cite{zhouQuantumApproximateOptimization2020,edwardfarhiQuantumApproximateOptimization2014,wurtzCounterdiabaticityQuantumApproximate2021}.  Making use of this relation, the parameter  initialization of QAOAs can be simplified associated with an improved performance  \cite{brandaoFixedControlParameters2018,yuQuantumApproximateOptimization2021,wurtzCounterdiabaticityQuantumApproximate2021}. In parallel to explore the initialization strategy, another crucial topic is designing advanced training strategies of QAOA to avoid local optima and accelerate optimization. Concrete examples include modifying the objective functions \cite{barkoutsosImprovingVariationalQuantum2020}, applying the iterative training strategy \cite{bravyiObstaclesStatePreparation2020,camposTrainingSaturationLayerwise2021}, and using adaptive mixing operators \cite{zhuAdaptiveQuantumApproximate2020,yuQuantumApproximateOptimization2021,hadfieldQuantumApproximateOptimization2019,wangXYMixersAnalytical2020}. Despite the remarkable achievements, little progress has been made in overcoming the scalability issue of QAOAs, whereas the ultimate goal of the most advanced QAOA is solving a problem with hundreds of vertices \cite{ebadiQuantumOptimizationMaximum2022}. The main challenges come from the fact that manipulating a graph with $n$-nodes requires $O(n)$ qubits but the most advanced quantum machines nowadays can only provide a very limited number of qubits with $n\approx 100$. Moreover, due to the a high-level noise and barren plateaus, QAOAs may suffer from the trainability issue for the large $n$  \cite{lotshawScalingQuantumApproximate2022,mccleanBarrenPlateausQuantum2018,leeProgressFavorableLandscapes2021,du2021learnability,zhang2021toward}, which degrade their performance~\cite{harriganQuantumApproximateOptimization2021a,marshall2020characterizing}. Although an initial attempt of the scalable QAOAs has been addressed by   \cite{liLargescaleQuantumApproximate2021,akshay2021parameter}, their approach encounters the sample complexity issue \footnote{The approach  proposed by  \cite{liLargescaleQuantumApproximate2021} breaks one graph into two subgraphs sharing common nodes. To sample a good candidate solution, the local solution of these common nodes should be exactly overlapped. In this respect, the sample complexity of their approach grows with number of common nodes which makes it harder to sample a good candidate solution.}. To this end, it still remains obscure whether QAOAs can outperform classical approaches towards large-scale combinatorial problems. 

\begin{figure}[t]
\includegraphics[width=0.48\textwidth]{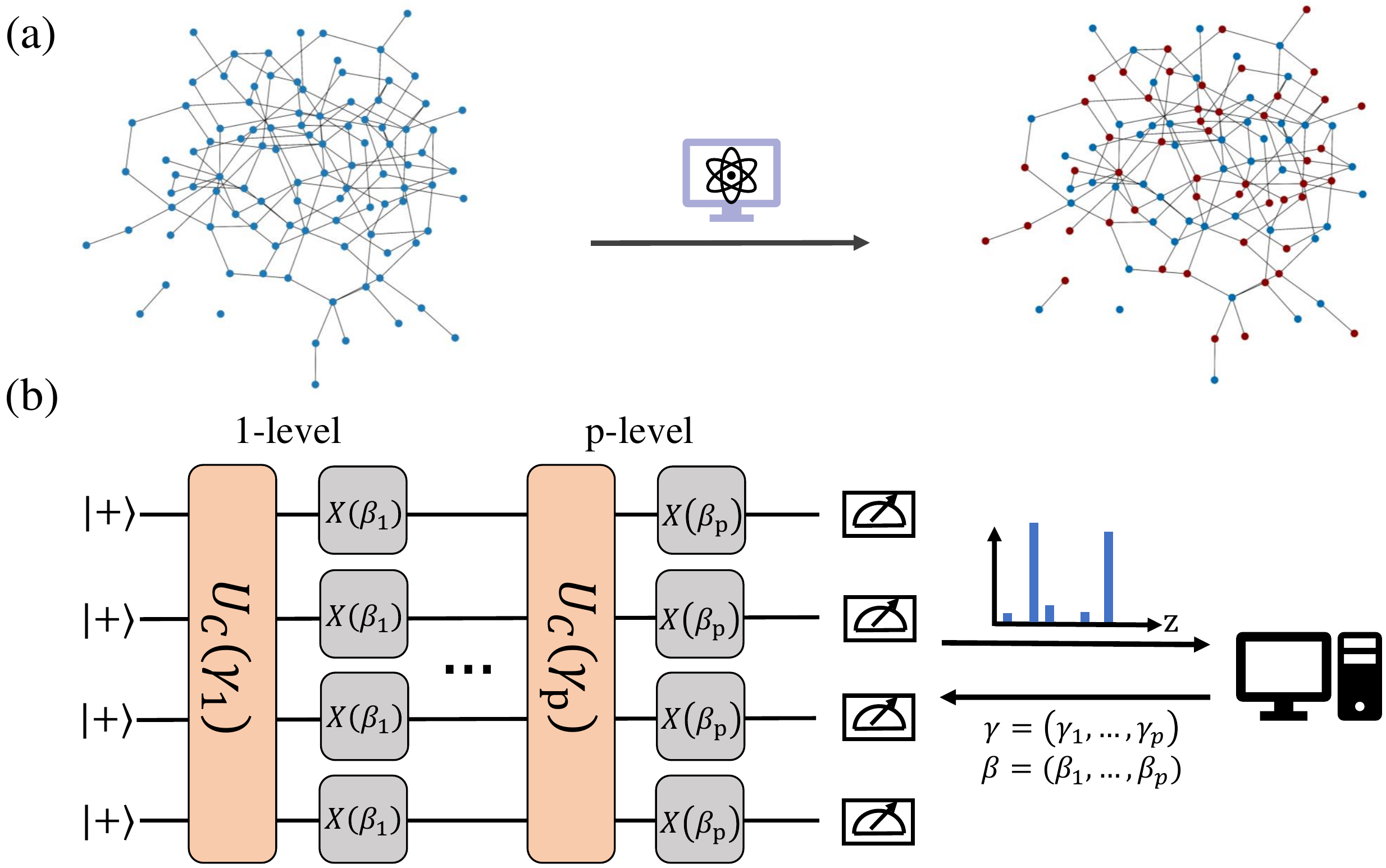}
\caption{\label{fig:1} \small{\textbf{MaxCut and QAOA} (a) One instance of MaxCut with 100 vertices. The left part is the problem graph. In the right part, two subsets of vertices are distinguished by its color as blue or red. (b) For a $p$-level QAOA,   $U_C(\gamma_i)$ and $U_B(\beta_i)$ are alternately applied to an initial state. The classical optimizer uses the measured bitstring to updated parameters of the circuit.}}
\end{figure}

To further enhance the capability of QAOAs, here we investigate the possibility of QAOAs toward solving large-scale MaxCut problems, especially when the problem size is greatly larger than the qubit count of the accessible quantum machines. In particular, we  revisit MaxCut through the lens of the divide-and-conquer heuristic, i.e., splitting the given graph into multiple subgraphs, seeking the solutions of these subgraphs in parallel,  and then merging these solutions to obtain the global solution. Notably, we prove that due to the inherent $\mathbb{Z}_2$ symmetry in MaxCut, the merging process can be cast to a new MaxCut problem. This observation hints that the large-scale MaxCut problem can be tackled by MaxCut solvers in a hierarchical way. To this end, we propose QAOA-in-QAOA  ($\QAOA$) to solve MaxCut with tens of thousands of nodes using NISQ machines. In addition, $\QAOA$ can integrate with the community detection method in the process of graph partition to attain better performance. On the theoretical side, we show that the lower bound of the approximation ratio of $\QAOA$ is $1/2$. On the experimental side, we first approximate the solution of MaxCut instances with 2000 vertices using $\QAOA$ executed on a 10-qubit quantum simulator. The achieved results are competitive or even better than the best known classical method. Moreover, through a systematical investigation, we verify that the density of graphs and subgraphs, the size of subgraphs and partition strategies are the decisive factors effecting the performance of $\QAOA$. These results suggest that $\QAOA$ provides a potential and novel way of utilizing  NISQ machines to solve practical learning problems with computational advantages.  

\begin{figure*}[htb]
\includegraphics[width=0.98\textwidth]{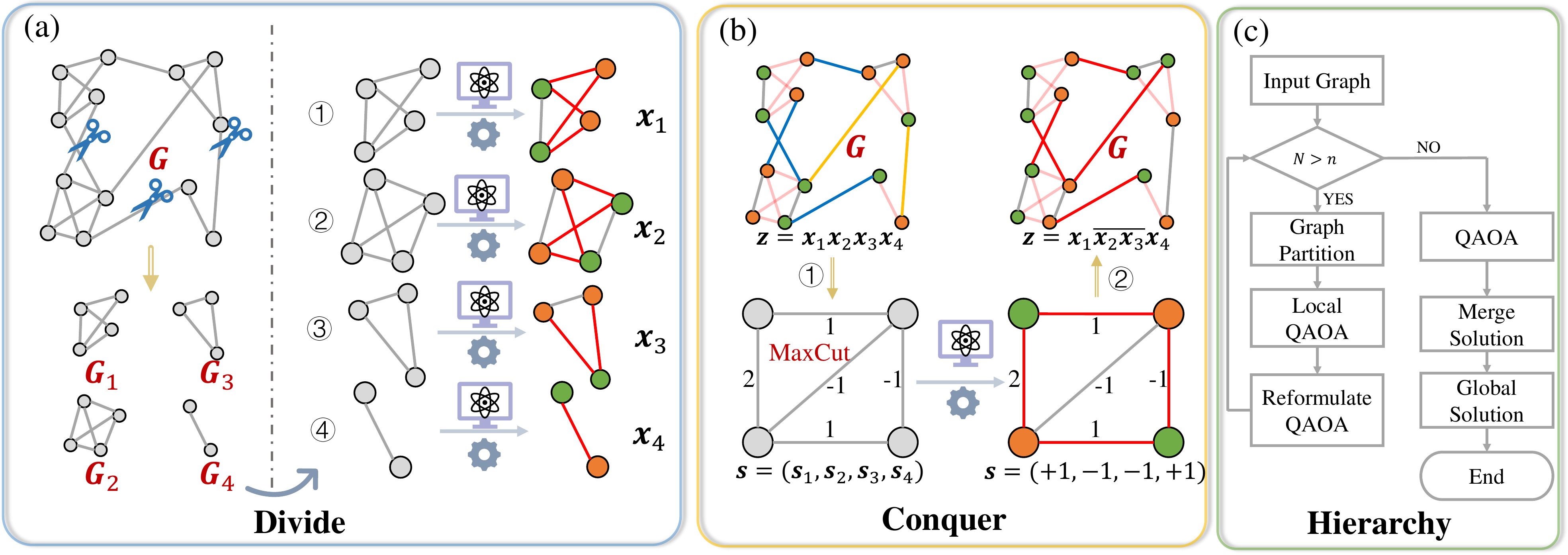}
\caption{\small{\textbf{Scheme of $\QAOA$.} (a) A graph is partitioned into four parts, where each one is no larger than the available number of qubits (e.g., $4$ qubits). Then the adopted MaxCut solvers are applied to all subgraphs in parallel. The green node refers to the bit $+1$ and the orange node refers to the bit $-1$. The cut edges are in red. (b) The first step highlighted by the brown arrow refers to merge local solutions of all subgraphs and calculate the value of cut between subgraphs. Yellow lines refer to `cut' and blue lines refer to `uncut'. The lower left plot indicates that the accommodation of local solutions can be reformulated a new Maxcut problem with four nodes. The first step highlighted by the brown arrow means that MaxCut solvers are applied to solve this new problem. Then, local solutions are merged according to solution of $s$. (c) An extremely large MaxCut can be solved by applying $\QAOA$ in a hierarchy way. When the graph size is above the limitation of qubits, it is partitioned and solved locally, and then reformed as a new MaxCut until its size is no larger than the available number of qubits. After all optimizations, low-level local solutions are merged according to high-level solutions.} }
\label{fig:2}
\end{figure*}

\section{Preliminaries}
The main focus of this study is solving MaxCut problems as shown in Fig.~\ref{fig:1}(a). Formally, let $\mathsf{G}(\mathsf{V}, \mathsf{E})$ be an undirected graph, where the number of vertices is $|\mathsf{V}|=N$ and the edge weight for $(i,j)\in \mathsf{E}$ is  $\bW_{ij}=W_{ji}$. Define a cut as a partition of the original set $\mathsf{V}$ into two subsets  $\mathsf{S}$ and $\mathsf{T}$ with $\mathsf{V}=\mathsf{S}\bigcup \mathsf{T}$ and $\mathsf{S}\bigcap \mathsf{T} = \emptyset$. The aim of MaxCut is to find $\mathsf{S}$ maximizing the sum of edge weight connecting vertices in $\mathsf{S}$ and $\mathsf{T}$. Denote $\bz_i=+1$ ($\bz_i=-1$)  when the $i$-th vertex is in $\mathsf{S}$   ($\mathsf{T}$), any partition of $\mathsf{V}$ can be represented by a bitstring $\bz \in \{+1,-1\}^n$. The optimal solution $\bz^*$ of MaxCut  maximizes the following objective function
\begin{equation}\label{eqn:def_maxcut}    
        C(\bz) := \frac{1}{2}\sum_{(i,j)\in \mathsf{E}} \bW_{ij}(1-\bz_i\bz_j) = c - \frac{1}{2}\sum_{(i,j)\in \mathsf{E}}\bW_{ij}\bz_i\bz_j,
\end{equation}
where $c=\frac{1}{2}\sum_{(i,j)\in \mathsf{E}} W_{ij}$ only depends on the problem and is independent of $\bm{z}$. Theoretical studies have proven that finding $\bz^*$ is NP-hard so in most cases we are searching for an approximation of $\bz^*$ \cite{karpReducibilityCombinatorialProblems1972}.  The best-known classical MaxCut solver is Goemans-Williamson (GW) algorithm \cite{goemansImprovedApproximationAlgorithms1995}, which uses the semi-definite programming (SDP) technique to ensure 0.879 approximation ratio \cite{goemans879approximationAlgorithmsMAX1994}. 

To carry out combinatorial problems on physical systems, it is necessary to   map the problem to a Hamiltonian whose ground state corresponds to the optimal solution \cite{gloverTutorialFormulatingUsing2019}. In the task of MaxCut with $n$ vertices, this Hamiltonian yields 
\begin{equation}
    \label{eq:2}
    H_C = \sum_{(i,j)\in \mathsf{E}}W_{ij}Z_iZ_j,
\end{equation}
where $Z_i$ refers to the Pauli-Z operator applied on the $i$-th qubit with $i\in[n]$ \cite{nielsenQuantumComputationQuantum2010}. The optimal solution $\ket{\bz^*}$ amounts to the computational basis state $|\bz\rangle$  minimizing $\langle \bz|H_C|\bz\rangle$. Since $H_C$ is diagonal, $\ket{\bz^*}$  also refers to its ground state. 

Quantum approximate optimization algorithm and its variants (QAOAs) \cite{edwardfarhiQuantumApproximateOptimization2014}, which absorb the merits of quantum annealing \cite{farhiQuantumSupremacyQuantum2019a} and variational quantum algorithms~\cite{cerezoVariationalQuantumAlgorithms2021}, are proposed to solve combinatorial problems on NISQ machines with potential advantages. As shown in Fig.~\ref{fig:1}(b), when applied to MaxCut, QAOA approximates the ground state by an ansatz state
\begin{equation}\label{eqn:def-QAOA}
    |\Phi(\bgamma, \bbeta)\rangle =  U_B(\bbeta_p)U_C(\bgamma_p)\dots U_B(\bbeta_1)U_C(\bgamma_1)|s\rangle
\end{equation}
where $|s\rangle=(|0\rangle^{\otimes n}+|1\rangle^{\otimes n})/\sqrt{2}$ is the initial state, $\bgamma, \bbeta \in (0, 2\pi]^p$ are variational parameters, and $U_B(\beta) = \exp(-i\beta \sum_k^nX_k)$ and $U_C(\gamma) = \exp(-i\gamma H_C)$. To find $\ket{\bz^*}$, a classical optimizer is used to update $\bgamma$ and $\bbeta$ by minimizing the following objective function 
\begin{equation}
    \label{eq:4}
    C(\bgamma,\bbeta)=\langle\Phi(\bgamma, \bbeta) | H_C|\Phi(\bgamma, \bbeta)\rangle.
\end{equation}
In the optimal setting, the optimal solution yields $\ket{\bz^*}=\ket{\Phi(\bgamma^*,\bbeta^*)}$ with $\ket{\Phi(\bgamma^*,\bbeta^*)}= \arg\min C(\bgamma,\bbeta)$.  

 \begin{figure*}[htb]
\includegraphics[width=0.98\textwidth]{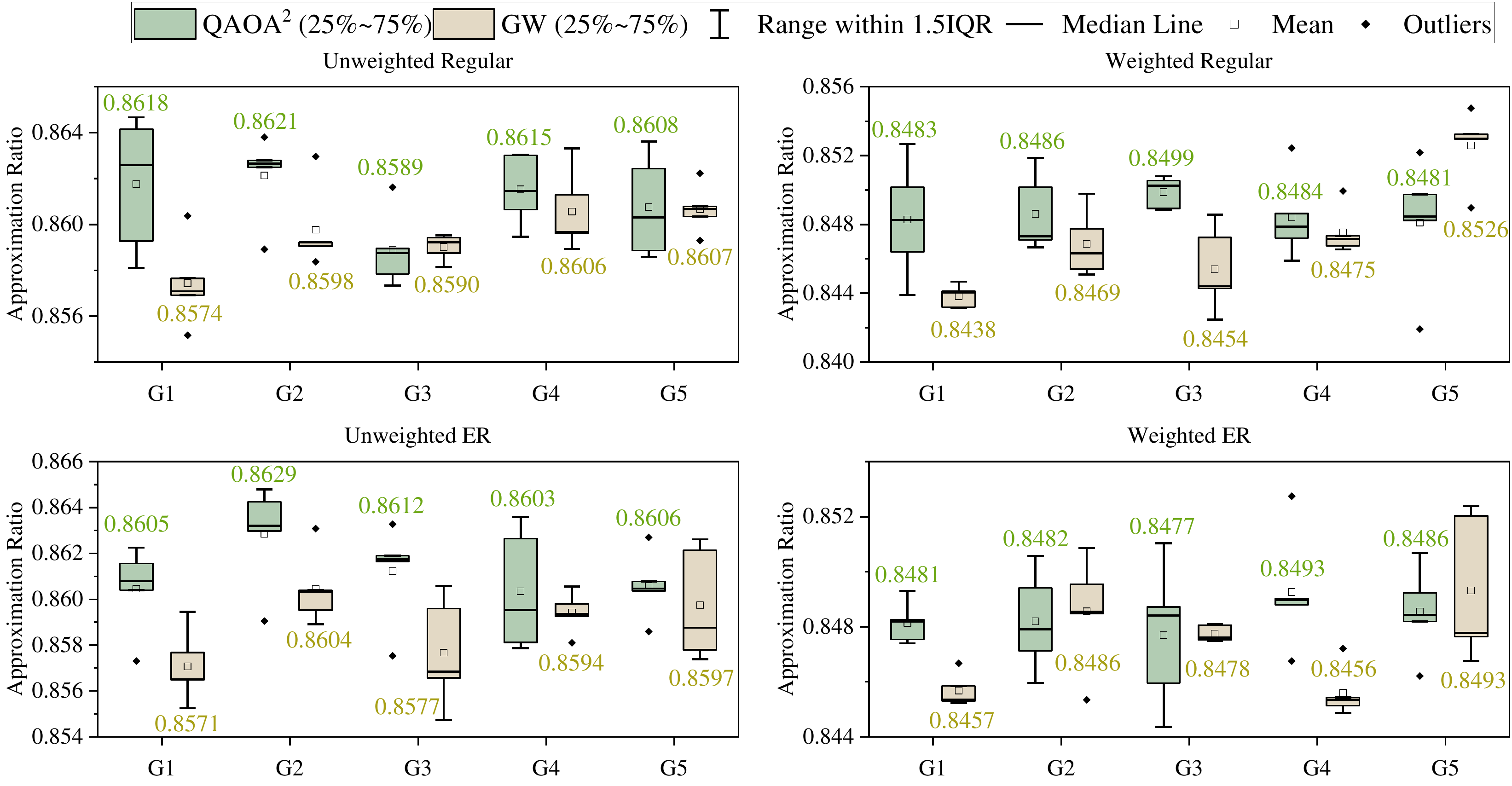}
\caption{ \small{\textbf{Approximation ration of $\QAOA$ and  GW over 2000-vertex graph instances under $C_{\text{SDP}}^*$.}  Numerical results on four types of graphs. There are five instances in each type. The number of vertices is fixed to  be $2000$. The x-axis represents the graph instance G1-G5 and the y-axis represents the approximation ratio. In each panel, $\QAOA$ is in green and GW is in yellow. The number beside each box is the mean value of approximation ratios in five trials, which is represented by the square symbol in each box. }}
\label{fig:3}
\end{figure*} 

\section{QAOA-in-QAOA}
Standard QAOAs require $n$ qubits to solve a MaxCut problem with $n$ vertices. This linear dependence suppresses the   power of QAOA, since the accessible quantum resources nowadays are vey limited. To assign the capability of QAOAs for solving large-scale problems, here we revisit MaxCut through the lens of the divide-and-conquer heuristic. Specifically, in the \textit{dividing} step, we partition the  given graph $\mathsf{G}$ into $h$ subgraphs $\{\mathsf{G}_i(\mathsf{V}_i,\mathsf{E}_i)\}_{i=1}^h$, where $\mathsf{V}=\bigcup_{i=1}^h \mathsf{V}_i$ and $\mathsf{V}_i\bigcap \mathsf{V}_j = \emptyset$ when $i\neq j$. An intuition of this partition process is exhibited in Fig.~\ref{fig:2}(a). Once the partition is completed,   the MaxCut solvers are exploited to seek optimal solutions of these subgraphs in parallel. We denote the optimized solutions for all $h$ subgraphs as $\{\bx_i\}_{i=1}^h$ with $\bx_i\in \{+1, -1\}^{|\mathsf{V}_i|}$ for $\forall i\in[h]$. Due to the $\mathbb{Z}_2$ symmetry in MaxCut \cite{bravyiObstaclesStatePreparation2020}, the bitstring  $\bbx_i$, which flips all bits in $\bx_i$,  also corresponds to the solution of $\mathsf{G}_i$ for $\forall i\in[h]$. In the \textit{conquering} step, as shown in Fig.~\ref{fig:2}(b), the obtained solutions of all subgraphs are merged to obtain the global solution $\bz$ of $\mathsf{G}$. Since there are two solutions for each subgraph, the total number of the possible global solutions is $2^h$, i.e.,  $\bz\in \mathcal{Z}:=\{\bx_1,\bbx_1\}\oplus \{\bx_2,\bbx_2\}... \oplus\{\bx_h,\bbx_h\}$. Taking into account the connections among $h$ subgraphs, the global solution yields
\begin{equation}\label{eqn:maxcut-divi-co}
	\hat{\bz}=\arg\max_{\bz\in \mathcal{Z}} C(\bz).
\end{equation}
The following theorem illustrates that seeking $\hat{\bz}$  can be cast into a new MaxCut problem, where the corresponding proof is provided in Appendix \ref{apx:aggregate}. 
\begin{thm}\label{thm:di-an-co-maxcut}
Suppose that the graph $\mathsf{G}$ is partitioned into $h$ subgraphs $\{\mathsf{G}_i\}_{i=1}^h$ and the optimized local solutions are $\{\bx_i\}_{i=1}^h$. To find the bitstring $\hat{\bz}$ in Eq.~(\ref{eqn:maxcut-divi-co}), let $\bs_i= +1$ (or $-1$) be the indicator of adopting $\bx_i$ (or $\bbx_i$). Then $\hat{\bz}$  is identified by $\bs=(\bs_1,\dots, \bs_h)\in \{\pm1\}^h$, i.e., 
\begin{equation}
    \label{eq:6}
    \max_{\bz\in \mathcal{Z}}C(\bz) \equiv \max_{\bs}   C^\prime(\bs)  = -\frac{1}{2}\sum_{i < j} w_{ij}^\prime \bs_i\bs_j + c,
\end{equation}
where $w_{ij}^\prime = \frac{1}{2}\sum_u \sum_v \bx_i^{(u)}\bx_j^{(v)} W_{iu,jv} $, $\bx_i^{(u)}$ is the $u$-th bit in $\bx_i$, $W_{iu,jv}$ is the weight of edge between the two nodes corresponding to $\bx_i^{(u)}$ and $\bx_j^{(v)}$, and $c=\sum_i C_i(\bx_i) + \sum_{i < j}\frac{1}{2}\sum_u \sum_v W_{iu,jv}$, $C_i(\bx_i)$ is the optimized value of cut for the subgraph $\mathsf{G}_i$. 
\end{thm}

\begin{figure}[htb]
\includegraphics[width=0.45\textwidth]{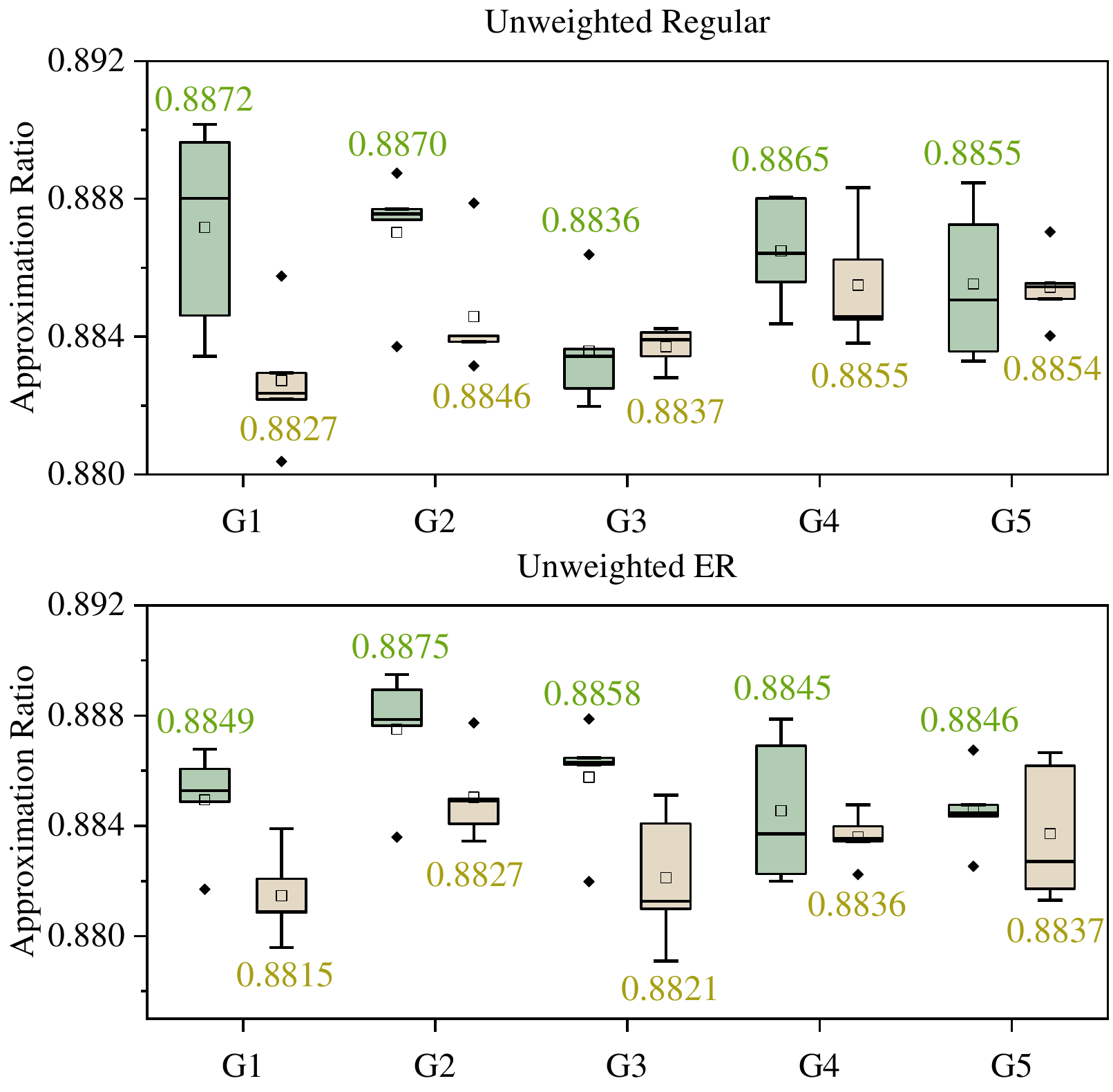}
\caption{ \small{\textbf{Approximation ration of $\QAOA$ and  GW over 2000-vertex graph instances under $C_{\text{asymp}}^*$.}  The labels follow the same meaning with those used in  Fig.~\ref{fig:3}. The only difference is the way of calculting the approximation ratio, where the denominator is replaced by $C^*_{\text{asymp}}$. }}
\label{fig:7}
\end{figure} 

The above results hint that MaxCut problems can be tackled in a hierarchical way. According to the reformulated MaxCut in the divide-and-conquer manner, we devise QAOA-in-QAOA, dubbed \textit{$\QAOA$}, which allows us to use an $n$-qubit quantum machine to solve an $N$-vertex MaxCut with $n\ll N$. The schematic of  $\QAOA$ is shown in Fig.~\ref{fig:2}. More specifically, in the partition procedure, the graph is divided into several subgraphs that are compatible with quantum devices. This step can be achieved by using random clustering, community detection,  or other advanced strategies (See Appendix \ref{apx:partition} for elaborations). The setting of $h$ is flexible, whereas the only requirement is that the size of subgraphs should be less than $n$, i.e., $|\mathsf{V}_i|\leq n$, $\forall i \in [h]$. After partitions, all subgraphs $\{\mathsf{G}_i\}_{i=1}^h$ are solved independently by QAOAs to collect $\{\bx_i\}_{i=1}^h$. Last, to obtain the estimated global solution $\hat{\bz}$,  we apply QAOAs again to optimize the merging through reformulated MaxCut according to Theorem~\ref{thm:di-an-co-maxcut}. Note that when $h>n$, the available number of qubits is insufficient to complete the merging process. As such, $\QAOA$ applies the partition procedure successively on $C^\prime(\bs)$ until the number of subgraphs is no larger than $n$. The diagram of $\QAOA$ is summarized in Fig.~\ref{fig:2}(c).

$\QAOA$ embraces two attractive theoretical advantages. First, compared with other QAOA solvers, $\QAOA$ is immune to the scalability issue. This characteristic is highly desired for NISQ machines, which provides the opportunity to attain potential quantum advantages towards large-scale combinatorial problems. Moreover, the following theorem whose proof is given in Appendix \ref{apx:proof} guarantees that the approximation ratio of $\QAOA$ is always better than the random guess. 
\begin{thm}\label{thm:low-bound}
	Following notations in Theorem \ref{thm:di-an-co-maxcut}, the estimated solution $\hat{\bz}$ output by  $\QAOA$ always satisfies 
	\begin{equation}
		C(\hat{\bz}) \geq \frac{1}{2}\sum_{(i,j)\in \mathsf{E}}W_{ij}.
	\end{equation}
\end{thm}

\section{Experiment Results}

\begin{figure*}[htb]
\includegraphics[width=0.98\textwidth]{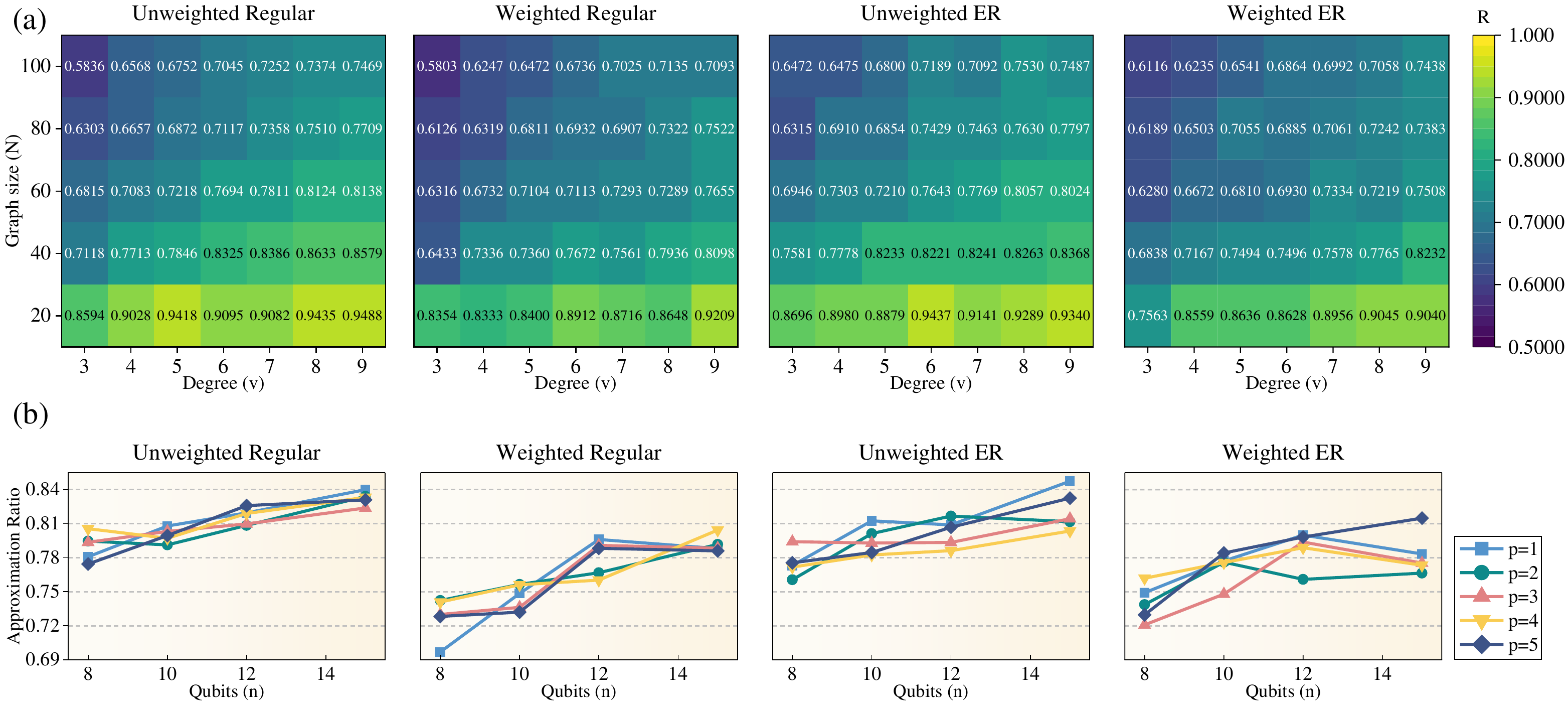}
\caption{  \small{(a) \textbf{Influence of graph density}. In heat map, x-axis represents the degree of vertices and the y-axis represents the size of graphs. The color of points indicates the approximation ratio of corresponding graph setting. Brighter color means higher ratio. (b) \textbf{ $\QAOA$ with different qubit size}. Fix graphs to 60-vertex and 9-degree. The performance grows as the qubit size grows from 8 to 15. Five types of curves represents the level of QAOA circuit from 1-5.  }}
\label{fig:4}
\end{figure*}

In this section, we conduct numerical simulation to evaluate the ability of $\QAOA$ towards large-scale MaxCut problems. Specifically, we first give the setup of our experiments, including hyper-parameter settings and the constructions of graphs for MaxCut.   Then we compare performance of our proposal with other MaxCut solvers. Last, we comprehend the potential factors influencing the performance of $\QAOA$.

To fully evaluate the performance of $\QAOA$, we collect a set of graphs varying in size, degrees of nodes, and the weight of edges. For convenience, an unweighted (or weighted) $d$-regular graph with $n$ nodes is abbreviated as ``u$d$r-$n$" (or ``w$d$r-$n$"). Similarly, for the Erdos-Renyi (ER) graph whose degree of each node is randomly assigned, an unweighted (or weighted) ER graph with the average degree being $d$ and the node size being $n$ is abbreviated as ``u$d$e-$n$" (or ``w$d$e-$n$"). For the weighted graphs, their weights are integers uniformly sampled from $[0,5]$. For each setting, we generate multiple instances to obtain the statistical results of the employed MaxCut solver.

When $\QAOA$ is applied, we specify that the allowable number of qubits is $n\leq 10$ and the random partition strategy is adopted. To find the local solutions $\bx_i$ in the dividing step, we use a 1-layer QAOA (i.e., $p=1$ in Eq.~(\ref{eqn:def-QAOA})) and the number of iterations for training QAOA is set as $T=20$.  To systematically investigate the potential of $\QAOA$ compared with the classical MaxCut solvers, we employ GW algorithm as the reference. The setting of hyper-parameters of adopted MaxCut solvers is given in Appendix \ref{append:sec:sim-res}.

In all simulations, we adopt the approximation ratio as the metric to compare the performance of different MaxCut solvers. Mathematically, given a MaxCut problem, denote $C_A$ as the cuts achieved by algorithm $\mathcal{A}$, the approximation ratio yields  
\begin{equation}
    r_{\mathcal{A}} = \frac{C_\mathcal{A}}{C^*}.
\end{equation}
where $C^*$ denotes the optimal value of MaxCut. Considering that the exact optimal value $C^*$ is exponentially difficult to get for large $N$, we substitute it with the optimal value $C^*_{\text{SDP}}$ of semidefinite programming optimization step in GW. Another metric we adopt for unweighted regular graphs and unweighted ER graphs with degree $d$ is an asymptotic value derived in \cite{demboExtremalCutsSparse2017}, where the optimal value of cut is $(\frac{d}{4} + P\sqrt{\frac{d}{4}}+o(\sqrt{d}))N$ and $P = 0.7632$ is the constant in Parisi formula. In calculating approximation rate, we set  the optimal result as $C^*_{\text{asymp}} = (\frac{d}{4} + P\sqrt{\frac{d}{4}})N$.

We first apply $\QAOA$ to 2000-vertex graphs, i.e., u100r-2000, w100r-2000, u100e-2000, and w100e-2000. The numerical results are shown in Fig.~\ref{fig:3} and Fig.~\ref{fig:7} under the measures $C^*_{\text{SDP}}$ and $C^*_{\text{asymp}}$, respectively. Each setting includes 5 graph instances. Meanwhile, for each instance, we use $5$ different random seeds to initialize parameters of $\QAOA$. In almost all instances, $\QAOA$ outperforms GW. For example, under the measure $C_{\text{SDP}}^*$, the averaged approximation ratio $\QAOA$ is higher than that of GW, except for the instance $G_3$ with u$1000$r-2000, the instance G5 with w$1000$r-2000, and the instances G2, G3, and G5 with w$1000$e-2000. Under the measure $C_{\text{asymp}}^*$, the averaged approximation ratio $\QAOA$ is higher than that of GW, except for the instance $G_3$ with u$1000$r-2000. In this  instances, the maximal differences of the approximation ratio is only $0.0001$.  Moreover, both $\QAOA$ and GW can obtain a better performance on unweighted graphs than weighted graphs.  A possible reason is the intrinsic hardness in finding optimal solution for weighted graphs. Besides, since the distribution of weighted edges is unbalanced and the graph connectivity is sparse, random partition used in $\QAOA$ is not a suitable choice, which may leave most edges remained among subgraphs. To this end, we exhibit how advanced partition methods, i.e., community detection algorithm, can further improve the power of $\QAOA$ in Appendix \ref{append:sec:sim-res}. Note that $\QAOA$ may attain a better runtime efficiency than GW, benefiting from its hierarchical scheme and the computational advantage of quantum algorithm.

\medskip
We next explore the potential factors that may influence the performance of $\QAOA$.  As mentioned previously, the sparsity of graph may reduce the power of $\QAOA$. To fully understand this effect, we conduct the systematical simulations on mild-size graphs varying in the number of nodes and the graph connectivity. Specifically, we study u$d$r-$n$, w$d$r-$n$, u$d$e-$n$ and w$d$e-$n$ with $d\in [3,9]$ and $n\in [20,100]$. For each setting, we generate 10 instances and use the average approximation ratio of them to evaluate the performance of $\QAOA$. Here both $C^*_{\text{SDP}}$ and $C^*_{\text{asymp}}$ are too loose, so we use the value of cut searched by GW algorithm as $C^*$ to calculate the approximation ratio $r_{\mathcal{A}}$. As shown in Fig.~\ref{fig:4}(a), an evident observation is that $\QAOA$ prefers denser graph than sparse graph. For example, $\QAOA$ achieves an approximation ratio of $0.9488$ on u$9$r-$20$ but only $0.5836$ on u$3$r-$100$. Consequently, a reasonable conjecture is that the main contribution of cuts comes from the cuts inside subgraphs. In order to improve the performance of $\QAOA$, one possible way is making the subgraphs as dense as possible. An alternative approach is adopting a better graph partition strategy such as the community detection method discussed in Appendix \ref{apx:partition}. 

Apart from the property of graph, the hyperparameter setting of $\QAOA$, i.e., the qubit counts $n$ and the number of layers $p$, may also effect its performance. With this regard, we consider the setting of graphs with 60 vertices and 9 degree i.e. u9r-60, u9e-60, w9r-60, and w9e-60. The number of qubits $n$ is chosen in $\{8, 10, 12, 15\}$ and the level $p$ is chosen in $[1,5]$. The results are collected across 10 instances for each setting. The achieved results are shown in Fig.~\ref{fig:4}(b). Specifically, under the measure of the averaged approximation ration, a deeper level $p$ slightly contribute much improvement. A concrete example is u9e-60, where the performance of $\QAOA$ with $p=4$ is inferior to the performance of $\QAOA$ with $p=3$.  Nevertheless,  the performance $\QAOA$ can increase significantly with the larger subgraph size $n$. For example, when the level is specified to be $p=1$, the approximation ratio of $\QAOA$ is increased by 0.06   in u9r-60 and 0.1 in w9r-60 when $n$ improves from 8 to 14. These observations indicate that executing $\QAOA$ on a large quantum system contributes to a better performance.

\section{Discussion}

In this study, we propose $\QAOA$ that utilizes the structure of graphs  and $\mathbb{Z}_2$ symmetry to solve large-scale MaxCut problems on small-scale quantum devices. We prove that a hierarchical scheme can be achieved via reformulated MaxCuts. The approximation ratio is always greater than 0.5. The numerical results show that the proposed $\QAOA$ achieves comparable performance on $2000$-vertex graphs against the best known classical algorithm. Moreover, numerical results indicates that   $\QAOA$ can attain better performance for denser graphs which are hard for conventional QAOAs.  Our work sheds light on solving large-scale problems with potential advantages in NISQ era.

There are several important future research directions. First, it is crucial to design more instance-aware partition strategies to further improve the capabilities of $\QAOA$. Furthermore, an intriguing direction is integrating $\QAOA$ with with distributed variational quantum optimization techniques  \cite{buhrmanDistributedQuantumComputing2003,cuomoDistributedQuantumComputing2020,duAcceleratingVariationalQuantum2021}, which allows us to accelerate the optimization and understand the power of $\QAOA$ on large-scale problems. Next, since the subgraphs are independent, the performance of $\QAOA$ could be enhanced by employing advanced and problem-specific local $\QAOA$ solvers \cite{zhuAdaptiveQuantumApproximate2020,wangXYMixersAnalytical2020,du2020quantum}. Last, the concept of decomposing Hamiltonian by its symmetric property used in $\QAOA$ can be generalized to boost other variational quantum algorithms. For example, in quantum chemistry, some proposals of variational eigensolvers have used the natural symmetry of some molecular to reduce the required number of  qubits \cite{bravyiTaperingQubitsSimulate2017,liuVariationalQuantumEigensolver2019,caoLargerMolecularSimulation2021}. In quantum machine learning, the concept of decomposing Hamiltonian by its symmetric property can be leveraged to design powerful Hamiltonian-based quantum neural networks with some invariant properties \cite{meyer2022exploiting,skolik2022equivariant}. In this way, these QNNs can attain better convergence and generalization \cite{Junyu2022dynamic,du2021efficient,huang2021information,abbas2020power,du2022theory}.

\newpage

\clearpage
\newpage
\appendix

\onecolumngrid
 
\begin{center}
\large{\textbf{Supplementary Material: ``QAOA-in-QAOA: solving large-scale MaxCut problems on small quantum machines''}}	
\end{center}
 
\medskip 
\onecolumngrid

The organization of Supplementary Material is as follows. In Appendix \ref{apx:partition}, we present how advanced partitioned methods can further improve the performance of $\QAOA$. In Appendix \ref{apx:aggregate}, we provide the proof of Theorem \ref{thm:di-an-co-maxcut}. Then in Appendix \ref{apx:proof}, we demonstrate the proof of Theorem \ref{thm:low-bound}. Subsequently, we discuss how $\QAOA$ relates to the Hamiltonian splitting method ind Appendix \ref{apx:physics}. Last, in Appendix \ref{append:sec:sim-res}, we exhibit the omitted simulation details and more simulation results of $\QAOA$.

\section{Graph partitioning}\label{apx:partition}

A crucial step in $\QAOA$ is partitioning graph into subgraphs.  We note that the way of partition is diverse. Here we present two possible partition methods, i.e., random partition and   community detection based partition. We leave the design of more advanced partition methods as the future work.

\medskip
\textbf{Random partition.} The algorithmic implementation of the random partition is as follows. Given the number of qubits $n$ and a graph with size $N$, random partition successively samples $n$ vertices as a subgraph without replacement until all $\lceil N/n \rceil$ subgraphs are collected.

We remark that for dense graphs, random partition promises a good learning performance of $\QAOA$ since the probability of an edge existing between arbitrary two vertex is high. In contrast, for graphs whose expected vertex degree is low, random partition may lead to an inferior performance. This is because the collected subgraphs may contain few edges and most edges remaining between subgraphs.

\medskip
\textbf{Community detection based partition.} According to the above explanations, a natural idea to enhance the power of $\QAOA$ is to maintain as many edges inside each subgraph as possible, which in turn suppresses the error incurred by  partition. With this regard, we introduce modularity \cite{newmanFindingEvaluatingCommunity2004} as a measure of the quality of partitioning. A mathematical definition of modularity is 
\begin{equation}
    Q = \frac{1}{2m}\sum_{ij}\left[A_{ij}-\frac{k_ik_j}{2m}\right]\delta(c_i,c_j),
    \label{eq:a1}
\end{equation} 
where $m$ is the sum of weights of all edges, $A_{ij}$ is the element of adjacency matrix, $k_i$ is the sum of weights of edges connected to vertex $i$, and $c_i$ is the community that the vertex $i$ is assigned. Intuitively, the term $\frac{k_ik_j}{2m}$ indicates the probability of edge existing between $i$ and $j$ in a randomized graph. When the fraction of edges within communities equals to a randomized graph, the quantity $Q$ will be zero. When $Q > 0.3$, it indicates a significant community structure \cite{clausetFindingCommunityStructure2004}. 

$\QAOA$ pursues a high modularity $Q$, where the connectivity in subgraphs is dense but the connectivity between different subgraphs is sparse. An algorithm searching for high modularity partition is referred to as \textit{community detection algorithm}, which regards subgraphs as communities. Several algorithms have been proposed to maximize modularity and to find community structure in graphs \cite{newmanFindingEvaluatingCommunity2004,newmanFastAlgorithmDetecting2004,clausetFindingCommunityStructure2004,blondelFastUnfoldingCommunities2008}. Here we consider the greedy modularity maximization algorithm  \cite{clausetFindingCommunityStructure2004}. In particular, starting with each vertex being only member of its own community, the algorithm joins the pair of communities that increases the modularity the most. This procedure is continuously conducted until no such pair exists or the termination condition is met. 

We benchmark the performance of $\QAOA$ with different partition methods in Appendix \ref{append:sec:sim-res}.

\section{\label{apx:aggregate} Proof of Theorem \ref{thm:di-an-co-maxcut}}

\begin{proof}[Proof of Theorem \ref{thm:di-an-co-maxcut}]
Although $\bx_i$ and $\bbx_i$ yield the same local objective value (i.e.,  $C_i(\bx_i)=C_i(\bbx_i)$), they may lead to a distinct  objective value (i.e.,  $C(\bx_1...\bx_i...\bx_h)\neq C(\bx_1...\bbx_i...\bx_h)$), because of the connection between $\mathsf{G}_i$ and other subgraphs $\{\mathsf{G}_{i'}\}_{i'\neq i}$. Considering that there are in total $2^h$ candidates  in $\mathcal{Z}$ in Eq.~(\ref{eqn:maxcut-divi-co}), our goal here is formulating an equivalent objective function that   finds the target bitstring satisfying
	$\hat{\bz}=\arg\min_{\bz\in \mathcal{Z}} C(\bz)$ using NISQ devices.   

Considering two neighboring subgraphs $\mathsf{G}_i$ and $\mathsf{G}_j$, we denote $\bs_i \in \{+1, -1\}$ as flipping indicator in the sense that $\bs_i=-1$ flips $\bx_i$ to be $\bbx_i$ and $\bs_i=+1$ keeps $\bx_i$ unchanged. When two subgraphs are synchronous, i.e. $\bs_i = \bs_j$, the inter-cut between them is retained and the inter-cut size is
\begin{equation}
    w_{ij}^{\text{sync}} = \sum_u \sum_v \frac{1}{2}(1-\bx_i^{(u)}\bx_j^{(v)})W_{iu,jv},
\end{equation}
where $\bx_i^{(u)}$ denotes the $u$-th bit in local solution bitstring $\bx_i$  with $u\in [|\mathsf{V}_i|]$ and $W_{iu,jv}$ is the weight of edge in $\mathsf{G}$  corresponding to two vertices $\bx_i^{(u)}$ and $\bx_j^{(v)}$.  Note that here we define $\bx_i \in \{+1,-1\}^{|V_i|}$, the coefficient $\frac{1}{2}$ acts as a standardization term to make sure $w_{ij}^{\text{sync}}$ is the sum of cut edges between $\mathsf{G}_i$ and $\mathsf{G}_j$ when adopting $\bx_i, \bx_j$ (or $\bbx_i, \bbx_j$).    When two subgraphs are asynchronous, i.e., $\bs_i = -\bs_j$, the inter-cut takes the form
\begin{equation}
    \begin{aligned}
        w_{ij}^{\text{async}} & =  \sum_u \sum_v \frac{1}{2}(1-\bx_i^{(u)}\bbx_j^{(v)})W_{iu,jv} \\
        &= \sum_u \sum_v \frac{1}{2}(1+\bx_i^{(u)}\bx_j^{(v)})W_{iu,jv},
    \end{aligned}
\end{equation}
where the second equality uses the relationship $\bbx_j = -\bx_j$.

Let $\bs=(\bs_1,\bs_2,\dots, \bs_h)\in \{+1,-1\}^h$ be an $h$-length bitstring as the indicator for the selection of $\{\bx_i,\bbx_i\}$ to form $\bz$. Define $w_{ij}^\prime = w_{ij}^{\text{async}}  - w_{ij}^{\text{sync}}$, $c=\sum_{i=1}^h C_i(\bx_i)+\sum_{i < j} \frac{1}{2}(w_{ij}^{\text{sync}} + w_{ij}^{\text{async}})$, and $C_i(\bx_i)$ as the local cut size of $\mathsf{G}_i$. For $\forall \bz\in\mathcal{Z}$, the objective value  $C(\bz)$ in Eq.~(\ref{eqn:maxcut-divi-co}) yields
\begin{equation}
    \label{eq:b3}
    \begin{aligned}
        &  C(\bz) \nonumber\\
        = & C(\bs_1\bx_1,\dots,\bs_h\bx_h) \nonumber\\
        = & \sum_{ i < j } \sum_u \sum_v \frac{1}{2}(1-\bs_i\bx_i^{(u)}\bs_j\bx_j^{(v)})W_{iu,jv} + \sum_i C_i(\bx_i) \nonumber\\ 
        = & \sum_{ i < j } \frac{1}{2}\left( -\bs_i \bs_j \sum_u \sum_v W_{iu,jv} \bx_i^{(u)}\bx_j^{(v)} + \sum_u \sum_v W_{iu,jv} \right) \nonumber\\ 
        & + \sum_i C_i(\bx_i) \nonumber\\
        = & \sum_{ i < j } \frac{1}{2}\left[(w_{ij}^{\text{sync}} + w_{ij}^{\text{async}}) + (w_{ij}^{\text{sync}} - w_{ij}^{\text{async}})\bs_{i}\bs_{j}\right] \nonumber\\
& + \sum_i C_i(\bx_i)\\
        = & -\frac{1}{2}\sum_{ i < j } w_{ij}^\prime \bs_i\bs_j + c,
    \end{aligned}
\end{equation}
where the second equality consists of two summations, i.e., the first term is the sum of inter cut between each pair of subgraphs and the second term is the sum of local cut inside subgraphs (which is not influenced by $\bs$),  the last second equality uses $w_{ij}^{\text{async}}  - w_{ij}^{\text{sync}} = \sum_u \sum_v W_{iu,jv} \bx_i^{(u)}\bx_j^{(v)}$ and $w_{ij}^{\text{async}}  + w_{ij}^{\text{sync}} = \sum_u \sum_v W_{iu,jv}$. 

Define $ C'(\bs) = c-\frac{1}{2}\sum_{ i < j } w_{ij}^\prime \bs_i\bs_j$. Since all local cuts $\{C_i(\bx_i)\}$  are fixed, we have
\begin{equation}
	\max_{\bs\in \{+1, -1\}^h}C'(\bs) = \max_{\bm{z}\in \mathcal{Z}}C(\bm{z}).
\end{equation} 
We find a good merging by optimizing the above.
\end{proof}

In most cases, optimizing merging of local solutions will improve the value of cut. However, if the local solutions happens to be in a good order, one can merge them without further optimization. We show the effect of merging optimization through experiments in Appendix \ref{append:sec:sim-res}.

\section{\label{apx:proof} Proof of Theorem \ref{thm:low-bound}}

\begin{proof}[Proof of Theorem \ref{thm:low-bound}]
	We follow notations defined in Appendix \ref{apx:aggregate} to prove  Theorem \ref{thm:low-bound}. Suppose that  $\mathsf{G}$ is partitioned into $h$ subgraphs $\mathsf{G}_1, \dots, \mathsf{G}_h$, we can divide edges into two parts $\mathsf{E}=\mathsf{E}^\text{inner}\bigcup \mathsf{E}^\text{inter}$ where $\mathsf{E}^\text{inner}=\bigcup_{i=1}^h \mathsf{E}_i$ denotes the set of edges inside all subgraphs and $\mathsf{E}^\text{inter}=\bigcup _{1\le i<j\le h}\mathsf{E}_{ij}$ denotes the set of edges between subgraphs.  Here the weight of edge $e\in \mathsf{E}$ is denoted   by $w(e)$. Then we have
\begin{equation}
    \label{eq:c1}
    \begin{aligned}
        \sum_{e \in \mathsf{E}} w(e) & = \sum_{e \in \mathsf{E}^\text{inner}} w(e) + \sum_{e \in \mathsf{E}^\text{inter}} w(e) \\
        & = \sum_{e\in \bigcup_{i=1}^h \mathsf{E}_i} w(e) + \sum_{e \in \bigcup _{1\le i<j\le h}\mathsf{E}_{ij}} w(e)\\
        & = \sum_{i=1}^h \sum_{e\in \mathsf{E}_i} w(e) + \sum_{1\le i<j\le h} \sum_{e\in \mathsf{E}_{ij}} w(e). 
    \end{aligned}
\end{equation}
When optimizing subgraphs, we can use any MaxCut solvers to return a set of local solutions $\{\bm{x}_i\}_{i=1}^h$ such that the cut value is greater than half of the sum of edge weights for all subgraphs $\{\mathsf{G}_i\}_{i=1}^h$. Mathematically, the sum of edge weight for each subgraph satisfies  
\begin{equation}
    \label{eq:c2}
    C_i(\bm{x}_i) \ge \frac{1}{2} \sum_{e\in \mathsf{E}_i} w(e) 
\end{equation}
and the sum of edge weight for all subgraphs yields 
\begin{equation}
    \label{eq:c3}
    C^{\text{inner}} = \sum_{i=1}^h C_i(\bm{x}_i) \ge \frac{1}{2} \sum_{e\in \bigcup_{i=1}^h \mathsf{E}_i} w(e).
\end{equation}
The above result means that we can always at least obtain half of the sum of edge weight edges inside $h$ subgraphs. 

Combining Eq.~(\ref{eq:c1}) and Eq.~(\ref{eq:c3}), an observation is that if $\QAOA$ outputs a solution $\bs=\{\bs_i\}_{i=1}^h$ such that the cut value of intra-subgraphs  achieves at least half of the second term  $\sum_{1\le i<j\le h} \sum_{e\in \mathsf{E}_{ij}} w(e)$ in Eq.~(\ref{eq:c1}), then the total cut value for the whole graph is greater than $\frac{1}{2}\sum_{e \in \mathsf{E}} w(e)$. Recall the terms $w_{ij}^{\text{async}}$ and $w_{ij}^{\text{sync}}$ defined in Theorem \ref{thm:di-an-co-maxcut}, we have 
\begin{equation}
    \label{eq:c4}
    \begin{aligned}
        \sum_{e \in \mathsf{E}^\text{inter}} w(e) &= \sum_{1\le i<j \le h} \sum_{e\in \mathsf{E}_{ij}} w(e) \\
        & = \sum_{1\le i<j \le h}  \sum_u \sum_v W_{iu,jv} \\
        & = \sum_{1\le i<j \le h} w_{ij}^{\text{async}}  + w_{ij}^{\text{sync}} \\
        & =\sum_{1\le i<j \le h} 2w_{ij}^{\text{sync}} + (w_{ij}^{\text{async}} - w_{ij}^{\text{sync}}) \\
        & = \sum_{1\le i<j \le h} 2w_{ij}^{\text{sync}} + w_{ij}^\prime \\
        & = \sum_{1\le i<j \le h} 2w_{ij}^{\text{sync}} + \sum_{1\le i<j \le h} w_{ij}^\prime.
    \end{aligned}
\end{equation}
Note that the reformulated MaxCut in Theorem \ref{thm:di-an-co-maxcut} in $\QAOA$ is 
\begin{equation}
    \begin{aligned}
        C^\prime(\bs) & = \sum_{ i < j } \frac{1}{2}\left[(w_{ij}^{\text{sync}} + w_{ij}^{\text{async}}) + (w_{ij}^{\text{sync}} - w_{ij}^{\text{async}})\bs_{i}\bs_{j}\right] \nonumber\\
        & + \sum_i C_i(\bx_i) \\
        & = \sum_{ i < j } w_{ij}^{\text{sync}} + \sum_{ i < j } (w_{ij}^{\text{async}} - w_{ij}^{\text{sync}})(1-\bs_{i}\bs_{j}) \\
        & + \sum_i C_i(\bx_i) \\
        & = \sum_{ i < j }  w_{ij}^\prime(1-\bs_{i}\bs_{j}) + \sum_{ i < j } w_{ij}^{\text{sync}} + C^{\text{inner}}
    \end{aligned}
\end{equation}
where $\sum_{ i < j }  w_{ij}^\prime(1-\bs_{i}\bs_{j})$ is optimized by the MaxCut solver and at least half of $\sum_{ i < j }  w_{ij}^\prime$ is cut. Let $\bs^*$ be the estimated solution, we have 
\begin{equation}
    \label{eq:c5}
    \begin{aligned}
        C^\prime(\bs^*) & \ge \frac{1}{2}\sum_{ i < j } w_{ij}^\prime + \sum_{ i < j } w_{ij}^{\text{sync}} + C^{\text{inner}} \\
        & = \frac{1}{2}\sum_{e \in \mathsf{E}^\text{inter}} w(e) + C^{\text{inner}} \\
        & \ge \frac{1}{2}\sum_{e \in \mathsf{E}^\text{inter}} w(e) + \frac{1}{2}\sum_{e \in \mathsf{E}^\text{inner}} w(e) \\
        & = \frac{1}{2}\sum_{e \in \mathsf{E}} w(e)
    \end{aligned}
\end{equation}
where the first inequality uses the result of reformulated MaxCut, the first equality uses the result of Eq.~(\ref{eq:c3}), the second inequality uses the result of Eq.~(\ref{eq:c4}), and the last equality uses Eq.~(\ref{eq:c1}).
\end{proof}

We end this section by illustrating when the lower bound is achieved, Consider the example of a four-vertex unweighted ring where $\mathsf{V}=\{1,2,3,4\}$ and $\mathsf{E} = \{(1,2),(2,3),(3,4),(1,4)\}$. Suppose we partition it into two subgraphs $\mathsf{G}_1$ and $\mathsf{G}_2$ with $\mathsf{V}_1=\{1,3\}$ and $\mathsf{V}_2=\{2,4\}$. So further if one of the local solutions is $(+1,-1)$, the final cut is 2 and the ratio is $\frac{1}{2}$ no matter how the global solution is merged. A smarter partition will be $\mathsf{V}_1=\{1,2\}$ and $\mathsf{V}_2=\{3,4\}$ where all edges are cut eventually. This case shows that the lower bound can be attained in worst case and effective partition strategy can alleviate this issue.

\section{\label{apx:physics} Relation with Hamiltonian splitting}

Many quantum computing tasks such as QAOA or VQE aims to find an eigenstate (ground state or most excited state) corresponding to a target eigenvalue of a given Hamiltonian $H$. 
We prepare a quantum state $|\psi(\theta)\rangle$ on parameterized circuit and measure it with $H$.
Then we feed $\langle\psi(\theta)| H|\psi(\theta)\rangle$ to an optimizer and update the parameters.
To measure states on qubit quantum computers, this Hamiltonian is represented in terms of Pauli words and each Pauli word can be written as tensor product of Pauli matrices i.e. $H = \sum_k \alpha_k P_k$. 
Thus we have 
\begin{equation}
    \label{eq:e1}
    \langle\psi(\theta)| H|\psi(\theta)\rangle = \sum_k \alpha_k \langle\psi(\theta)| P_k|\psi(\theta)\rangle .
\end{equation}
This allows us to measure each Pauli words individually and add them together.
In QAOA for MaxCut problem, $H$ takes the form of \ref{eq:2} where all Pauli words are tensor products of two Pauli-Z matrices corresponding to edges in problem graph.

In general, an arbitrary Hamiltonian of $n$-qubit has $O(4^n)$ Pauli words so the query of circuit grows exponentially with qubit counts if we simply measure one Pauli word at a time.
Note that two observables can be measured simultaneously if they are commutable. Moreover, if a group of pairwise commuting observables share the same eigenbasis that diagonalizes them all simultaneously, they can be also measured on the same prepared state.
For example in MaxCut problem, all Pauli words share the same eigenbasis (computational basis) so we don't need to measure $O(n^2)$ terms individually but only once. 
One method to reduce circuit query is to split Hamiltonian into several clusters and each cluster is a commuting group of Pauli words we mentioned above.
The less number of clusters, the less of circuit query.
Recent researches mapped Hamiltonian splitting task into MinCliqueCover problem \cite{gokhaleMinimizingStatePreparations2019,verteletskyiMeasurementOptimizationVariational2020}.
Besides, additional speedup can be introduced by distributed quantum computing on multiple untangled quantum computers \cite{duAcceleratingVariationalQuantum2021}.

\begin{figure*}[h]
\centering
\includegraphics[width=0.75\textwidth]{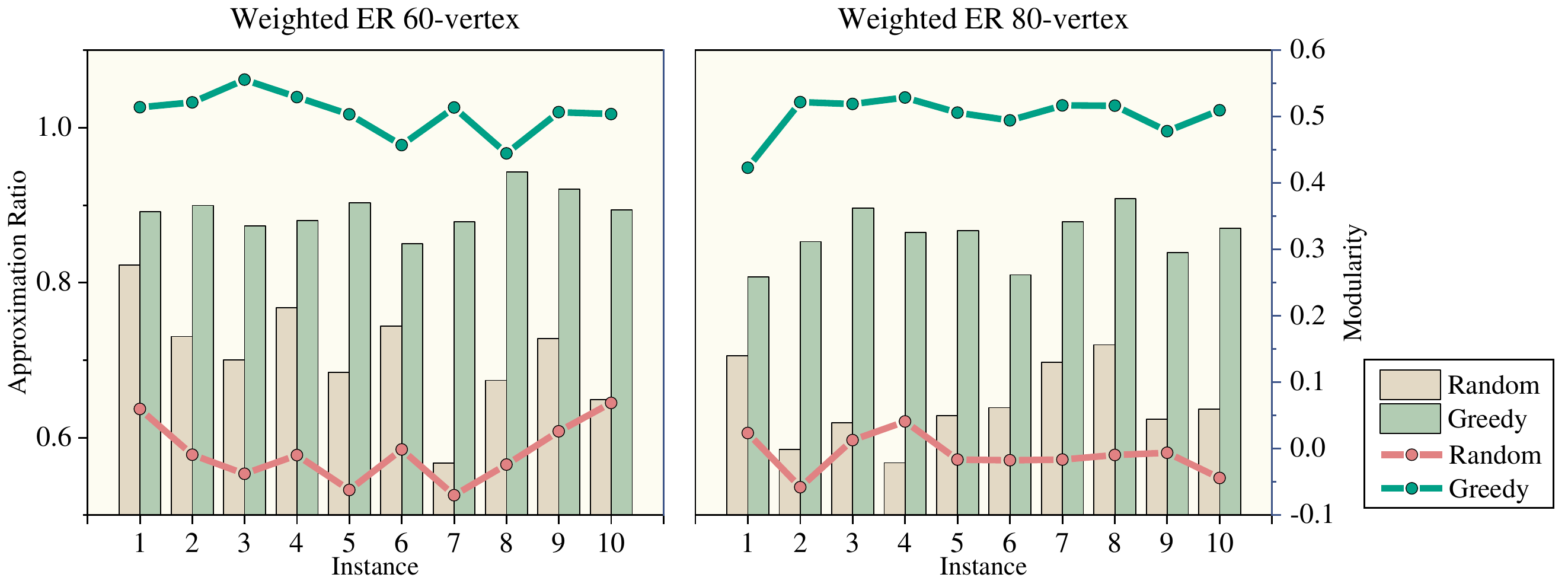}
\caption{ \small{\textbf{Results between random and greedy partition}. } The yellow bars and green bars refer to approximation ratios of two partition strategies. The dotted line refers to the modularity of partition of each graph, which measures the quality of partitioning.}
\label{fig:5}
\end{figure*}

Nevertheless, this does not reduce the required qubit counts on quantum computer. In order to reduce qubit counts, we need to ensure that any two clusters do not share common qubit i.e. Hamiltonian of each cluster cannot be used to measure the same qubit. Since there are always Pauli words with small coefficients which play little role in final $\langle\psi(\theta)| H|\psi(\theta)\rangle$, one can pretend they don't exists when constructing circuit thus partition qubits into several parts.
If so, we can built several small and independent circuit and measure each cluster on distributed quantum computers with less qubits.
The partitioning may follow the property of primal problem such as graph weight or the mutual information between clusters \cite{zhangVariationalQuantumEigensolver2021}.
To minimize the performance loss introduced by partitioning, one can use dressed Hamiltonian in measuring \cite{zhangVariationalQuantumEigensolver2021} or include a fixing step as we discussed in Appendix \ref{apx:aggregate}.

\section{Details of numerical simulations}\label{append:sec:sim-res}

\subsection{Implementation details of $\QAOA$}
\textbf{Implementation details of $\QAOA$.}
The QAOA used in $\QAOA$ are implemented by Pennylane \cite{bergholmPennyLaneAutomaticDifferentiation2020}. In  optimization, $\QAOA$ adopts the vanilla gradient descent method to update the trainable parameters of QAOAs in which the learning rate is set as 0.01. The number of shots of each circuit is set as 1000 to approximate the expectation value of the measurements in calculating gradients. In the process of sampling solutions, $\QAOA$ runs the optimized circuit for $1000$ times to return the same number of bitstrings as solution candidates. To be more specific, in each run, all $n$ qubits are measured by the computational basis along the $Z$ direction and the measured result of each qubit is either +1 or -1. Thus an $n$-dimensional bitstring $z \in \{+1,-1\}^n$ is sampled in one query of the circuit. We collect 1000 such bitstrings as candidates and select the one as solution whose corresponding eigenvalue is the smallest among the candidates.

The source code of $\QAOA$ is available at the Github Repository \url{https://github.com/ZeddTheGoat/QAOA_in_QAQA}.

\textbf{GW algorithm.}  GW maps the primal integer programming to continuous vector space and optimizes it with semidefinte programming. The binary solution is obtained by projecting vectors to a plane. Our SDP solver is implemented by CVXPY \cite{diamondCVXPYPythonEmbeddedModeling2016}, which uses SCS (Splitting Conic Solver)  \cite{odonoghueConicOptimizationOperator2016} following conventions. Here we adopt the default parameters in executing GW, where the max number of iterations is 2500, the convergence tolerance is $1\text{e}^{-4}$,  the relaxation parameter is 1.8, and the balance between minimizing primal and dual residual is 5.0. After optimization, the solution is projected to 100 random vectors and rounded to bitstrings. 

\subsection{More simulation results of $\QAOA$}

\textbf{Performance of $\QAOA$ with the advanced partition methods.} We investigate how the advanced partition method, i.e., the greedy community detection introduced in Appendix \ref{apx:aggregate}, effects the performance of $\QAOA$. To do so, we apply $\QAOA$ to two types of graphs, i.e., w$3$e-$60$ and w$3$e-$80$. To collect the statistical results, we generate 10 instances for each setting. The allowable number of qubits is $n < 15$. The implementation of the community detection method follows the algorithm proposed in \cite{clausetFindingCommunityStructure2004}, which is realized by NetworkX library \footnote{https://github.com/networkx/networkx}. 

\begin{figure*}[h]
\centering
\includegraphics[width=0.98\textwidth]{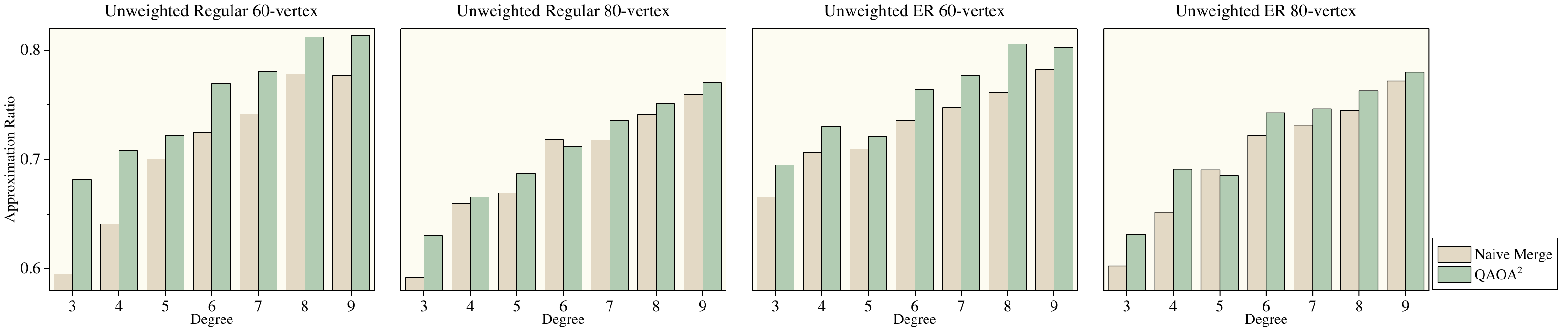}
\caption{ \small{\textbf{Results between Naive merging and $\QAOA$}. } The yellow bars represents the approximation ratio of $\QAOA$, where solutions are merged naively. The green bars represents the approximation ratio of $\QAOA$ introduced in the main text. }
\label{fig:6}
\end{figure*}

 The simulation results are showin in  Fig.~\ref{fig:5}, where the performance of $\QAOA$ is dramatically increased when the random parition is replaced by the community detection. Specifically, for the 7-th instance of w3e-60 and the 4-th instance of w30e-80, the approximation ratio is increased by more than 0.3 compared to the random partition. Uder the measure of the modularities defined in Eq.~\ref{eq:a1}, the achieved results of $\QAOA$ with the community detection strategy are all above 0.4. This indicates good partition, which implies that most edges are kept within subgraphs. By contrast, the modularities of $\QAOA$ with the random partition strategy are about 0. These observations accord with our conjecture in main text that the cut mainly comes from within subgraphs. The advanced partition strategy, e.g., the community detection, ensures the subgraphs as dense as possible and the number of edges left between subgraphs is minimized.

\textbf{Effect of merging optimization.} We next elucidate the importance of recasting the merging process as a new MaxCut problem used in $\QAOA$. Particularly, we conduct an ablation study by evaluating the performance of $\QAOA$ when the merging process is replaced by a naive heuristic, i.e.,  flipping non-local solutions as $\bs=\mathbf{1}$. To do so, we compare these two merging methods on u$d$r-60, u$d$r-80, u$d$e-60 and u$d$e-80 with $d\in [3,9]$ and 10 instances per setting. The number of qubits is set as $n< 10$. The collected results are shown in Fig.~\ref{fig:6}. For almost all settings, $\QAOA$ outperforms the naive approach, except for u6r-80 and u5e-80.  One possible reason is that local solutions is already good enough.

\end{document}